\documentclass[ aps,pra,showkeys, superscriptaddress, twocolumn,nofootinbib]{revtex4-1}
\usepackage{amssymb,amsmath,amscd}
\usepackage{hyperref}
\usepackage{amsmath, amsthm, amscd, amsfonts, amssymb,mathtools}
\usepackage{xcolor}
\usepackage{ulem}

\newtheorem{thm}{Theorem}
\newtheorem{lemma}{Lemma}

\newtheorem{corollary}{Corollary}

\newcommand{\tr}{\mathrm{Tr }}

\newcommand{\bra}[1]{\langle {#1}|}
\newcommand{\ket}[1]{| {#1} \rangle}

\newcommand{\braket}[2]{\langle #1 | #2 \rangle}

\newcommand{\ens}{\mathcal{E}(r_1,\cdots,r_m)}
\newcommand{\pro}{\mathcal{P}(r_1,\cdots,r_m)}

\begin{document}

\title{A Structure of \\ Minimum Error Discrimination for \\ Linearly Independent States}

\author{Tanmay Singal}
\email{tanmaysingal@gmail.com}
\affiliation{Department of Applied Mathematics ,\\ Hanyang University, Ansan
Kyunggi-do, Korea}
\author{Eunsang Kim}
\email{eskim@hanyang.ac.kr}
\affiliation{Department of Applied Mathematics ,\\ Hanyang University, Ansan
Kyunggi-do, Korea}
\author{Sibasish Ghosh}
\email{sibasish@imsc.res.in}
\affiliation{Optics \& Quantum Information Group, The Institute of Mathematical Sciences, CIT Campus, Taramani, Chennai, 600 113, India}
\affiliation{Homi Bhabha National Institute, Training School Complex, Anushakti Nagar, Mumbai 400094, India}
\begin{abstract}
In this paper we study the Minimum Error Discrimination problem (MED) for ensembles of linearly independent (LI) states.  We define a bijective map from the set of those ensembles to itself and we show that the Pretty Good Measurement (PGM) and the  optimal measurement for the MED are related by  the map. In particular, the fixed points of the map are those ensembles for which the PGM is the optimal measurement. Also, we simplify the optimality conditions for the measurement of an ensemble of LI states.\end{abstract}


\maketitle

\section{Introduction}

In quantum state discrimination, one wishes to optimally ascertain which of a collection of states has been provided. In general, two parties, Alice and Bob, are involved in this scenario. We may formulate the discrimination problem in the following way. Let $\mathcal{H}$ be a $d$-dimensional Hilbert space. Alice prepares a quantum state $\rho_i$, from an ensemble of quantum states $\mathrm{P}=\{p_i, \rho_i\}_{i=1}^m$ with a priori probability $p_i$. Here the quantum states $\rho_i$ are density operators on $\mathcal{H}$ (i.e., $\rho_i \ge 0$, and $\mathrm{Tr} \, \rho_i = 1$ for all $1 \le i \le m$), and the a priori probabilities $p_1,\cdots, p_m$ are such that $p_i>0$ and $\sum_{i=1}^m p_i=1$. We assume that the basis vectors of $\mathrm{Range} \, \rho_i$ collectively span $\mathcal{H}$. Alice sends her state $\rho_i$ to Bob, without telling him what $i$ is. In order to find the value of $i$, Bob has to probe the state $\rho_i$ using an appropriate measurement. When the $\rho_i$'s are non-orthogonal, then they can't be perfectly distinguished. The average probability of error in his inference of the value of $i$ is $ \displaystyle \sum_{\substack{i,j=1 \\  i \neq j}}^{m} p_i \mathrm{Tr}\rho_i E_j$. Bob's objective is to obtain the positive operator valued measure (POVM), $\{E_i\}_{i=1}^m$,  which maximizes the probability of success,  i.e.,
\begin{equation}
\label{Ps}
p_s=  \underset{{\left\{E_i\right\}_{i=1}^{m}}}{\mathrm{Max}} \;  \sum_{i=1}^m p_i \tr \rho_i E_i,
\end{equation}
subject to the conditions $E_i\ge 0$ for all { $1 \le i \le m$} and $\sum_{i=1}^m E_i=\text{Id}$, where the maximum is taken over the set of all $m$-element POVMs. This optimization problem is known as Minimum Error Discrimination(MED), or the quantum hypothesis testing problem \cite{YKM75} \cite{Belavkin75}, \cite{BelavkinMaslov1988}, \cite{Helstrom76}. The POVM for which one obtains the maximum value is called the optimal POVM.

While there are many algorithms to iteratively solve the MED problem \cite{Helstrom82,JRF02,J09}, there are only a few ensembles of states for which closed-form expressions for the optimal POVMs and success probability have been obtained. Some prominent examples for these are the two state ensemble \cite{Belavkin75, Holeo1973}, ensembles of geometrically uniform states \cite{Belavkin75, Ban97,Sasaki97}, various ensembles of states for $\mathrm{dim \ } \mathcal{H} = 2$ \cite{Bae13, Ha13, Weir18}, etc.  In \cite{YF00}, Eldar and Forney considered a variant of the state discrimination problem for pure states, wherein the objective is to minimize the sum of the normed square of the distance between the pure states and corresponding measurement basis vectors. While most of the earlier results directly employ the optimality conditions (see Section \ref{sec:opt}) to solve the problem, some of the later results use a variety of different structures of the problem to solve it, for instance, the geometric structure of the problem \cite{Bae13, Bae2013, Ha13, Ha15}, and an algebraic structure \cite{Weir18}. { Some notable recent results include the following: exact analytic expressions for the optimal measurement strategies for trine states with arbitrary probabilities \cite{WHBC18}, algorithmically realizing the optimal measurement as a set of nested binary measurement  \cite{RPMG17}, employing results from group theory and representation theory to obtain the optimal success probability for geometrically uniform sets of states \cite{KGDS15}, etc. Many reviews on Quantum State discrimination can also be found in the literature, for instance see - \cite{C00, BHH04, BC09, B10,BK15}. }

A structure of the MED problem was discovered by V. P. Belavkin \cite{Belavkin75}. He showed that for each distinct optimal POVM for the MED of some ensemble $\mathrm{P} = \left\{ p_i, \rho_i \right\}_{i=1}^m$ of quantum states, one can find another ensemble of quantum states $\mathrm{Q}$, such that the pretty good measurement (PGM) of $\mathrm{Q}$ is the  optimal POVM of $\mathrm{P}$. In \cite{Mochon2006} it was shown that in the case of linearly independent pure states, one can relate $\mathrm{P}$ and $\mathrm{Q}$ by a bijective mapping. In this work we prove that such a bijective mapping exists on sets of ensembles of LI mixed states as well. Using this map one may solve the MED problem for LI mixed state ensembles. However to construct the map we need the optimal POVM, and hence without knowing the optimal POVM we cannot construct this map. { Our main results in this paper are as follows:  (1) we construct the inverse map explicitly. (2) We find the necessary and sufficient conditions for an ensemble of LI mixed states to be fixed points of this map (Theorem \ref{fix}). The fixed points of the map are ensembles whose optimal POVMs are their PGMs. Thus if an ensemble is a fixed point of this map, its optimal POVM is readily known. Thus we solve the MED problem for this class of ensembles. This is a generalisation of a result in \cite{Masahide98}. (3) In the course of inverting the map, we show that the optimality conditions for the MED of LI mixed state ensembles is actually simpler than the well-known optimality conditions (Theorem \ref{simple1}, Corollary \ref{cor}). This generalises a known result for LI pure states \cite{Helstrom82, Pozza15} to LI mixed states.}

This paper is organized as follows. In Section \ref{sec:opt} we give a brief summary of the optimality conditions for MED. In Section \ref{sec:Bela} we describe a structure of the MED problem which was introduced by Belavkin \cite{Belavkin75, BelavkinMaslov1988}. In Section \ref{sec:LI} we build on this structure to prove the existence of a map on the set of LI ensembles, such that the PGM of the image (under the map) is the optimal POVM of the pre-image. Also, at the end of this Section we show that the optimality conditions for MED of LI mixed ensembles is actually simpler than for the well-known optimality conditions for general ensembles of states. In Section \ref{R'} we prove that this map is bijective and explicitly construct its inverse. In Section \ref{sec:fix} we obtain necessary and sufficient conditions for the fixed points of this map. Section \ref{sec:last} concludes the paper.

\section{Optimality Conditions}
\label{sec:opt}

The set of all $m$-element POVMs is a convex set. Thus MED is a convex optimization problem. Thus, one can formulate the dual problem as follows: for a given ensemble $\mathrm{P} = \left\{ p_i, \rho_i \right\}_{i=1}^{m} $ of quantum states, find an operator $Z$ which minimizes $\tr Z$, subject to the condition $Z\ge p_i\rho_i$ for all $1 \le i \le m$. For the MED problem there is no duality gap and the dual problem can be solved to obtain the optimal POVM \cite{Mochon2006},\cite{EldrMegretskiVerg2003}, i.e.,
\begin{equation}
\label{duality}
p_s = \underset{Z \ge p_i \rho_i}{\mathrm{Min}} \; \mathrm{Tr}Z.
\end{equation}
We call the pair $(\{\Pi_i\}_{i=1}^m,Z)$  an {\it optimal dual pair} when $\{\Pi_i\}_{i=1}^m$ is an optimal POVM and $Z$ satisfies the duality \eqref{duality}.
For an optimal dual pair $(\{\Pi_i\}_{i=1}^m,Z)$ we have
\begin{equation}\label{duality2}
\sum_{i=1}^m p_i\tr(\Pi_i\rho_i) =\tr Z.
\end{equation}
{ The optimality conditions on the optimal POVM $\{ \Pi_i \}_{i=1}^{m}$ are given in the following theorem. For proofs we refer the reader to \cite{Belavkin75}, \cite{Helstrom76}, \cite{YKM75} and \cite{BC09}.}

\begin{thm}
\label{thm:optimality}
{ 
For an ensemble $\mathrm{P}=\{p_i,\rho_i\}_{i=1}^m$, an $m$-{\rm POVM} $\{\Pi_i\}_{i=1}^m$
is optimal if and only if it satisfies the following relations {\rm 1} and {\rm 2}: for all $i,j\in\{1,\cdots, m\}$,

\begin{enumerate}
\item[{\rm 1}.] \begin{equation}\label{dual4}
\Pi_j\left(p_j\rho_j-p_i\rho_i\right)\Pi_i=0,
\end{equation} {
\noindent or, equivalently,
\begin{equation}\label{dual1}
(Z-p_i\rho_i)\Pi_i=\Pi_i(Z-p_i\rho_i)=0,
\end{equation}
where
\begin{align}\label{dual2}
Z&=\sum_{i=1}^m p_i\Pi_i\rho_i=\sum_{i=1}^mp_i\rho_i\Pi_i.
\end{align}}
\item[{\rm 2}.] \begin{equation}\label{dual5}
Z\ge p_i\rho_i \Longleftrightarrow \sum_{j=1}^m p_j\rho_j\Pi_j-p_i\rho_i \ge 0,
\end{equation}
\hfill$\square$
\end{enumerate}}
\end{thm}

\noindent In \cite{Bae2013} it was established that the operator $Z$ from the optimal dual pair is unique, whereas the optimal POVM $\left\{ \Pi_i \right\}_{i=1}^{m}$ may not be unique.

\section{A Structure for the MED Problem}
\label{sec:Bela}

{ In this section, we study a mathematical structure for any general ensemble of states. This structure was first presented in \cite{BelavkinMaslov1988}. In Section \ref{sec:LI}, we show how this structure is modified when the states are LI.}

Let $(\{\Pi_i\}_{i=1}^m,Z)$ be the optimal dual pair for an ensemble $\mathrm{P}=\{p_i,\rho_i\}_{i=1}^m$ of quantum states. We construct an ensemble of quantum states associated with  $(\{\Pi_i\}_{i=1}^{m},Z)$. Let us define
\begin{equation}\label{sigma}
\sigma_i := \frac{Z\Pi_i Z}{\tr(Z^2\Pi_i)}, \; \text{for all} \; 1 \le i \le m,
\end{equation}
and
\begin{equation}\label{qii}
q_i :=\frac{\tr(Z^2\Pi_i )}{\tr(Z^2)},  \; \text{for all} \; 1 \le i \le m.
\end{equation}
{Since $\{ \Pi_i \}_{i=1}^{m}$ is a POVM},  $\sum_{i=1}^m q_i=1$ and
\begin{equation}\label{qisi}
q_i\sigma_i=\frac{Z\Pi_i Z}{\tr(Z^2)},  \; \text{for all} \; 1 \le i \le m.
\end{equation}

\begin{lemma}\label{lemma1}
Let $\mathrm{P}=\{p_i,\rho_i\}_{i=1}^m$ be an ensemble of quantum states and let $(\{\Pi_i\}_{i=1}^{m},Z)$ be an
optimal dual pair for $\mathrm{P}$. Then the ensemble $\mathrm{Q}=\{q_i,\sigma_i\}_{i=1}^m$ of quantum states defined in {\rm (\ref{sigma})} and {\rm (\ref{qii})} satisfy the following properties:
\begin{enumerate}
\item[(a)] $\sigma_i\ge 0$ for all $i=1,\cdots, m$,
\item[(b)] $\mathrm{Tr}(\sigma_i)=1$ for all $i=1,\cdots, m$,
\item[(c)] $\mathrm{Range} \; q_i\sigma_i \subseteq \mathrm{Range} \; p_i\rho_i$ for all $i=1,\cdots, m$.
\end{enumerate}
\end{lemma}

{\it Proof.} Conditions (a) and (b) follow directly from equations \eqref{sigma} and \eqref{qii}. By (\ref{dual5}), the operator $Z$ is invertible and by the definition of $\sigma_i$, we get $\text{rank }\sigma_i=\text{rank }\Pi_i$ and from equations (\ref{dual4}) we obtain that $Z \Pi_i Z = p_i \rho_i \Pi_i Z = p_i^2 \rho_i \Pi_i \rho_i$, for all $1 \le i \le m$. This implies that $\mathrm{Range} \left( q_i\sigma_i\right) \subseteq \mathrm{Range}\left( p_i\rho_i\right)$ for all $ 1 \le i \le m$. \hfill\qed



Let $\mathrm{Q}=\{q_i,\sigma_i\}_{i=1}^m$ be an ensemble of  quantum states such that $\sigma = \sum_{i=1}^m q_i \sigma_i > 0$. Then, the PGM of $\mathrm{Q}$ is defined as follows: for each $i\in\{1,\cdots,m\}$, let\footnote{One can define a PGM for an arbitrary ensemble of states $Q = \left\{q_i, \sigma_i \right\}_{i=1}^m$ using equation \eqref{PGMQ}. This is also true when $\mathrm{supp} \; \sigma$ is strictly smaller than $\mathcal{H}$. In such cases, we restrict the space to $\mathrm{span} \left\{ \mathrm{supp} \; \sigma_i \right\}_{i=1}^{m}$ to define $\sigma^{-1/2}$. We will see (in Theorem \ref{thm:Bela} and from Section \ref{sec:LI} onwards) that we only employ $\mathrm{Q}$ for which $\mathrm{supp} \; \sigma  = \mathcal{H}$, and hence $\sigma$ is always invertible on $\mathcal{H}$.}
\begin{equation}\label{PGMQ}
E_i :=\sigma^{-1/2}(q_i\sigma_i)\sigma^{-1/2}.
\end{equation}
Then it is easy to see that for all $i\in\{1,\cdots,m\}$,  $E_i\ge 0$ and
\begin{align*}
\sum_{i=1}^m E_i&=\sum_{i=1}^m \sigma^{-1/2}(q_i\sigma_i)\sigma^{-1/2}\\
&=\sigma^{-1/2}\sum_{i=1}^mq_i\sigma_i\sigma^{-1/2}\\
&=\sigma^{-1/2}\sigma\sigma^{-1/2}=\text{Id}.
\end{align*}
Thus  we see that $\{E_i\}_{i=1}^m$ is a POVM. 
{
\begin{thm}[\cite{BelavkinMaslov1988}]
\label{thm:Bela}
Let $\mathrm{P}=\{p_i,\rho_i\}_{i=1}^m$ be an ensemble of quantum states with an optimal dual pair $(\{\Pi_i\}_{i=1}^m,Z)$, and let  $\mathrm{Q}=\{q_i,\sigma_i\}_{i=1}^m$ be the ensemble constructed from the optimal dual pair using equations \eqref{sigma} and \eqref{qii}.
Then $\{ \Pi_i \}_{i=1}^m$ is the $\mathrm{PGM}$ of $Q$.
\end{thm}}
\begin{proof} From \eqref{qisi} we get,
\begin{equation}
\label{sigmaZ}
\sigma = \sum_{i=1}^{m} q_i\sigma_i = \dfrac{Z^2}{\tr(Z^2)},
\end{equation}
and thus $\sigma^{-1/2} = \sqrt{\tr(Z^2)} Z^{-1}$. Now from (\ref{qisi}) and (\ref{PGMQ}), we get for all $ 1 \le i \le m$
\begin{align*}
E_i&=\sigma^{-1/2} \left( q_i \sigma_i \right) \sigma^{-1/2}\\
&=\sqrt{\tr(Z^2)} Z^{-1}
\frac{Z\Pi_i Z}{\tr(Z^2)}\sqrt{\tr(Z^2)} Z^{-1}\\
&=\Pi_i.
\end{align*}
\end{proof}
This shows that the PGM of $\mathrm{Q}$ is the optimal POVM for MED of $\mathrm{P}$. In particular, in the case of pure states we have a nice property which is proved in \cite{Belavkin75, Mochon2006}.

\begin{thm}
\label{thm:Mochon}
Let $\mathrm{P}=\{p_i, |\psi_i\rangle\langle\psi_i|\}_{i=1}^m$ be an ensemble of pure states on a $d$-dimensional Hilbert space $\mathcal{H}$ and let $\{E_i\}_{i=1}^m$ be the PGM of the pure state ensemble $\mathrm{Q}=\{q_i, |\psi_i\rangle\langle\psi_i|\}_{i=1}^m$. For all $i=1,\cdots,m$, if  $p_i\langle\psi_i|\rho_q^{-1/2}|\psi_i\rangle =C$, where $\rho_q = \sum_{i=1}^{m} q_i \ket{\psi_i}\bra{\psi_i}$ and $C$ is a constant so that $\sum_{i=1}^m p_i=1$, then $\{E_i\}_{i=1}^m$ is the optimal POVM for $\mathrm{P}$.
\hfill\qed
\end{thm}

\section{Structure for Linearly Independent States of the MED Problem}
\label{sec:LI}

{ In this section we show the following: in Theorem \ref{thm:Bela}, when the states in $\mathrm{P}$ are LI and mixed, then $\mathrm{P}$ is {\it mapped} to $\mathrm{Q}$. This result is a generalization of part of a result in \cite{Mochon2006}, where it was derived for the LI pure state case.}



Consider an ensemble $\mathrm{P}=\{p_i,\rho_i\}_{i=1}^m$ of quantum states on an $d$-dimensional Hilbert space $\mathcal{H}$. Assume that the eigenvectors of $\rho_i$, $1\le i\le m$ collectively span $\mathcal{H}$. Since each density operator $\rho_i$ is { h}ermitian,  it has the eigendecomposition as {$\rho_i=\sum_{k=1}^{r_i} \lambda_{ik} \ket{\phi_{ik}}\bra{\phi_{ik}}$, where $\braket{\phi_{ik}}{\phi_{ik'}} = \delta_{kk'},$ for $1 \le k, k' \le r_i$. Thus, $\text{Rank }\rho_i=r_i$.} The set of quantum states $\{\rho_i\}_{i=1}^m$ is said to be {\it linearly independent}  if the set of vectors $\{|\phi_{ik}\rangle \mid 1\le k\le r_i, \ 1\le i\le m\}$ are linearly independent. Since the set $\{|\phi_{ik}\rangle \mid 1\le k\le r_i, \ 1\le i\le m\}$ spans $\mathcal{H}$, we have $\sum_{i=1}^mr_i=d$. An ensemble $\mathrm{P}=\{p_i,\rho_i\}_{i=1}^m$ of quantum states is said to be a {\it LI  state ensemble} if the set $\{\rho_i\}_{i=1}^m$ of density operators form a linearly independent set.

Define $\mathcal{E}(r_1,\cdots,r_m)$ to be the set of all LI  state ensembles $\mathrm{P}=\{p_i,\rho_i\}_{i=1}^m$  such that $\text{Rank }\rho_i=r_i$ for all $ 1\le i \le m$. In \cite{Eldar2003}, it was shown that for each element in $\mathcal{E}(r_1,\cdots,r_m)$, the optimal POVM is a projective measurement.  More explicitly, we have

\begin{thm}[\cite{Eldar2003}]\label{thm:Eldar03}   Let $\mathrm{P}=\{p_i,\rho_i\}_{i=1}^m\in \mathcal{E}(r_1,\cdots,r_m)$. Then the optimal POVM { for $\mathrm{P}$, i.e.,}  $\{\Pi_i\}_{i=1}^m$ is a projective measurement. In other words, it satisfies
$\Pi_i\Pi_j=\delta_{ij}\Pi_i$, $\Pi_i^\dag = \Pi_i$,  for all $i,j\in \{1,\cdots,m\}$, and $\sum_{i=1}^m \Pi_i=\mathrm{Id}$. { Also, $\text{\rm Rank}(\rho_i)=\text{\rm Rank}(\Pi_i)$ for all $i=1,\cdots, m$}. \hfill\qed
\end{thm}

Furthermore we have,

\begin{thm}\label{thm:uniquePOVM}
Let $\mathrm{P}=\{p_i,\rho_i\}_{i=1}^m\in \mathcal{E}(r_1,\cdots,r_m)$. {Then its optimal POVM is unique.}
\end{thm}

{\it Proof. } {Let $\{\Pi_i\}_{i=1}^m$ and $\{\overline{\Pi}_i\}_{i=1}^m$ be two optimal POVMs} for $\mathrm{P}=\{p_i,\rho_i\}_{i=1}^m$. Then
\[\text{Rank}(\Pi_i)=r_i=\text{Rank}(\overline{\Pi}_i) \ \ \text{ for all }i=1,\cdots, m.\]
By (\ref{dual2}),
\[Z=\sum_{i=1}^m p_i\Pi_i\rho_i=\sum_{i=1}^mp_i\overline{\Pi}_i\rho_i.\]
By the result in \cite{Bae2013}, the operator $Z$ is unique. Since $\rho_i$ are linearly independent we have $\Pi_i=\overline{\Pi}_i$ for all $i=1,\cdots,m$.\hfill\qed

Let $\mathcal{P}(r_1,\cdots,r_m)$ be the set of all $m$-element projective measurements $\{\Pi_i\}_{i=1}^m$ such that $ \text{Rank}(\Pi_i)=r_i$ for all $ 1 \le i \le m$.

Define the followimg map $\mathbf{OP}$, which we call the {\it optimal POVM map}
\[\mathbf{OP}: \mathcal{E}(r_1,\cdots,r_m)\longrightarrow \mathcal{P}(r_1,\cdots,r_m)\] as follows: for each $\mathrm{P}=\{p_i,\rho_i\}_{i=1}^m\in \mathcal{E}(r_1,\cdots,r_m)$,
\[\mathbf{OP}(\mathrm{P})=\{\Pi_i\}_{i=1}^m \in \mathcal{P}(r_1,r_2,\cdots,r_m),\]
where $\{\Pi_i\}_{i=1}^m$ is the optimal POVM for the ensemble $\mathrm{P}$. {Note that Theorem \ref{thm:uniquePOVM} guarantees that $\mathbf{OP}$ is a {\it well-defined} map}.

We also  define a map
\[\mathfrak{R} : \mathcal{E}(r_1,\cdots,r_m)\longrightarrow\mathcal{E}(r_1,\cdots,r_m)\] as follows: for each $\mathrm{P}=\{p_i,\rho_i\}_{i=1}^m\in \mathcal{E}(r_1,\cdots,r_m)$, let $\mathbf{OP}(\mathrm{P})=\{\Pi_i\}_{i=1}^m$. Then
as constructed in Section \ref{sec:Bela}, we have $\mathfrak{R}(\mathrm{P})=\mathrm{Q}\in\mathcal{E}(r_1,\cdots,r_m)$, where $\mathrm{Q}=\{q_i,\sigma_i\}_{i=1}^m$ and $q_i\sigma_i=\dfrac{Z\Pi_i Z}{\tr(Z^2)}$. {Note that  $\mathfrak{R}$ is well-defined}.

Using the pretty good measurement one can also define PGM as a function \[\mathrm{PGM}:\mathcal{E}(r_1,r_2,\cdots,r_m)\longrightarrow\mathcal{P}(r_1,\cdots,r_m)\]
such that
\[\mathrm{PGM}(\mathrm{Q})=\{\Pi_i\}_{i=1}^m, \text{ where }\Pi_i =\sigma^{-1/2} \left( q_i \sigma_i \right) \sigma^{-1/2}.\]

We have defined two  functions $\mathbf{OP}$ and  $\mathrm{PGM}$ from the  set $\mathcal{E}(r_1,r_2,\cdots,r_m)$ to a set $\mathcal{P}(r_1,\cdots,r_m)$ and $\mathfrak{R}$ maps from $\mathcal{E}(r_1,r_2,\cdots,r_m)$ to itself. The relation between these three functions are given in the following theorem.
\begin{thm}
\label{thm5}
\label{thm:rela}
Let $\mathcal{E}(r_1,\cdots,r_m)$ be the set of LI states ensemble whose $i$-th  state is of rank $r_i$ and let
$\mathcal{P}(r_1,\cdots,r_m)$ be the set of projective POVMs such that $\mathrm{rank}(\Pi_i)=r_i$ for all $i=1\cdots,m$. Then we have the following relation
\begin{equation}
\label{mainresult}
\mathbf{OP}=\mathrm{PGM}\circ \mathfrak{R}.
\end{equation}

\hfill\qed
\end{thm}

\begin{proof}
{ Theorem \ref{thm:uniquePOVM} implies that $\mathbf{OP}$ and $\mathfrak{R}$ are well-defined maps. Then, Theorem \ref{thm:Bela} implies that equation \eqref{mainresult} is true.} 
\end{proof}

{ Note that, in general, one can't define the maps $\mathbf{OP}$ and $\mathfrak{R}$ because the optimal dual pair for $\mathrm{P}$ may not be unique, unless one is restricted to a case like the LI states.}

Moreover, we can show that the map $\mathfrak{R}$ is bijective. For this we first explicitly construct another function $\mathfrak{R}'$ on $\ens$, and later show that $\mathfrak{R}'$ is the left and right inverse of $\mathfrak{R}$, i.e., we show that $\mathfrak{R}^{-1}$ exists and it is equal to $\mathfrak{R}'$.

\section{Bijectivity of $\mathfrak{R}$}
\label{R'}

In order to show that the map $\mathfrak{R}: \mathcal{E}(r_1,\cdots,r_m)\longrightarrow\mathcal{E}(r_1,\cdots,r_m)$ is bijective, we construct the inverse of the map.

Let $\mathrm{Q}=\{q_i,\sigma_i\}_{i=1}^m$ be any element in $\mathcal{E}(r_1,\cdots,r_m)$, $\sigma = \sum_{i=1}^{m} q_i \sigma_i$ and let $\mathrm{PGM}(\mathrm{Q})=\{\Pi_i\}_{i=1}^m \in \mathcal{P}\left(r_1,\cdots,r_m\right)$, then for all $i=1,\cdots, m$, { $\Pi_i$ is given by the RHS in equation \eqref{PGMQ}.}

Consider the following decomposition of $\sigma^{1/2}$.

\begin{align}
\label{sigma12}
  \sigma^{1/2}  = & \;\left( \mathrm{Id} - \Pi_i + \Pi_i \right) \sigma^{1/2} \left( \mathrm{Id} - \Pi_i + \Pi_i \right) \notag \\
 = & \; \Pi_i \sigma^{1/2} \Pi_i + \left( \mathrm{Id} - \Pi_i \right) \sigma^{1/2} \left( \mathrm{Id} - \Pi_i \right) \notag \\ & + \Pi_i \sigma^{1/2} \left( \mathrm{Id} - \Pi_i \right) + \left( \mathrm{Id} - \Pi_i \right) \sigma^{1/2} \Pi_i .
\end{align} { Choose an orthonormal basis in which $\Pi_i$ and $\mathrm{Id}-\Pi_i$ are simultaneously diagonal.} In such a basis, $\sigma^{1/2}$ can be represented by the following matrix
\begin{equation}
 \label{sigma12inbasis}
 \sigma^{1/2} \longleftrightarrow \begin{pmatrix}
                  A_i &  B_i \\
                  B_i^\dag & C_i
                 \end{pmatrix}.
\end{equation}
Since $\sigma^{1/2} > 0 $, $\left( \begin{smallmatrix} A_i & B_i \\ B_i^\dag & C_i \end{smallmatrix} \right) > 0$. Note that
\begin{itemize}
 \item[(a)] $A_i$ is the $r_i\times r_i$ block matrix within $\left(\begin{smallmatrix}
                  A_i &  B_i \\
                  B_i^\dag & C_i
                 \end{smallmatrix}\right)$, and hence $A_i > 0$,
 \item[(b)] $C_i$ is the $(d-r_i)\times(d-r_i)$ block matrix within $\left(\begin{smallmatrix}
                  A_i &  B_i \\
                  B_i^\dag & C_i
                 \end{smallmatrix}\right)$, and hence $C_i > 0$, and
 \item[(c)] $B_i$ is the $r_i\times (d-r_i)$ block matrix within $\left(\begin{smallmatrix}
                  A_i &  B_i \\
                  B_i^\dag & C_i
                 \end{smallmatrix}\right)$.
\end{itemize}
\noindent Define
\begin{equation}
\label{Schurc}
\Delta_i \equiv C_i - B_i^\dag \left(A_i\right)^{-1} B_i.
\end{equation}
\noindent Note that $\Delta_i$ is the Schur complement of $A_i$ in $\left( \begin{smallmatrix} A_i & B_i \\ B_i^\dag & C_i \end{smallmatrix} \right)$.  From \cite{Boyd}\footnote{See Appendix A.5.5, page 651 in \cite{Boyd}.}, we see that when $\left( \begin{smallmatrix} A_i & B_i \\ B_i^\dag & C_i \end{smallmatrix} \right) > 0$, the Schur complement of $A_i$ in $\left( \begin{smallmatrix} A_i & B_i \\ B_i^\dag & C_i \end{smallmatrix} \right)$ is also strictly positive. 
\noindent Define $X_i$ to be an operator, which is represented by the following matrix using the same basis as in \eqref{sigma12inbasis}
\begin{align}
 \label{X_i}
 X_i \longleftrightarrow &  \begin{pmatrix}
                  A_i &  B_i \\
                  B_i^\dag & C_i
                 \end{pmatrix} - \begin{pmatrix}
                  0 &  0 \\
                  0 & \Delta_i
                 \end{pmatrix} = \begin{pmatrix}
                  A_i &  B_i \\
                  B_i^\dag & B_i^\dag A_i^{-1} B_i
                 \end{pmatrix} \notag  \\
                = & \begin{pmatrix}
                  \mathrm{Id}_{r_i} &  0 \\
                  B_i^\dag A_i^{-1} & \mathrm{Id}_{d-r_i}
                 \end{pmatrix} \begin{pmatrix}
                  A_i &  0 \\
                  0 & 0
                 \end{pmatrix} \begin{pmatrix}
                  \mathrm{Id}_{r_i} & A_i^{-1} B_i \\
                  0 & \mathrm{Id}_{d-r_i}
                 \end{pmatrix} .
\end{align}
Thus we see that
\begin{equation}
\label{SchurRank}
\mathrm{Rank} \, X_i = \mathrm{Rank} \, A_i  = r_i \text{ \ and  \ } X_i \ge 0
\end{equation}
Now define
\begin{equation}\label{pi}
 p_i \equiv \dfrac{\tr X_i}{\sum_{j=1}^m \tr X_j} \text{ \ and \ } \rho_i \equiv \dfrac{X_i}{\tr X_i}.
\end{equation}
Thus we obtain the ensemble $\mathrm{P} = \left\{p_i, \rho_i\right\}_{i=1}^{m}$ of quantum states and by \eqref{SchurRank}  $\mathrm{Rank} \, \rho_i = r_i$, for all $ i=1,\cdots,m$.

\begin{thm}
\label{thm1}
For any $\mathrm{Q}=\{q_i,\sigma_i\}_{i=1}^m\in \mathcal{E}(r_1,\cdots,r_m)$, define $\mathfrak{R}'(\mathrm{Q})=\mathrm{P}$, where $\mathrm{P}=\{p_i,\rho_i\}_{i=1}^m$ is an ensemble of quantum states as given in {\rm \eqref{pi}}. Then $\mathrm{P}\in \mathcal{E}(r_1,\cdots,r_m)$ and $\mathfrak{R}'$ defines a function on $\mathcal{E}(r_1,\cdots,r_m)$. Furthermore,
$\mathrm{PGM}(\mathrm{Q})$ is the optimal POVM for MED of $\mathrm{P}$.

\end{thm}

\begin{proof}
Let $\mathrm{PGM}(\mathrm{Q})=\{\Pi_i\}_{i=1}^m\in \mathcal{P}(r_1,\cdots,r_m)$ and define
\begin{equation}
 \label{ZQ}
Z \equiv \dfrac{\sigma^{1/2}}{\sum_{j=1}^{m}\mathrm{Tr}X_j}.
\end{equation}
Then for each $i=1,\cdots,m$,
\begin{align*}
Z-p_i\rho_i&=\frac{\sigma^{1/2}}{\sum_{j=1}^{m}\mathrm{Tr }X_j}-\frac{X_i}{\sum_{j=1}^{m}\mathrm{Tr}X_j}\\
&=\frac{1}{\sum_{j=1}^{m}\tr X_j}\left(\sigma^{1/2}-X_i\right).
\end{align*}
In the matrix representation used earlier we see that
\begin{equation}
 \label{st}
 \left( \sigma^{1/2} - X_i \right)\Pi_i \longleftrightarrow \begin{pmatrix}
                                           0 & 0 \\
                                           0 & \Delta_i
                                          \end{pmatrix}. \begin{pmatrix}
                                           \mathrm{Id}_{r_i} & 0 \\
                                           0 & 0
                                          \end{pmatrix} = \begin{pmatrix}
                                           0 & 0 \\
                                           0 & 0
                                          \end{pmatrix}.
\end{equation}
Thus $\left\{ \Pi_i \right\}_{i=1}^{m}$ and $Z$ satisfy the equation \eqref{dual1} for the ensemble $\left\{p_i, \rho_i \right\}_{i=1}^{m}$. Also, since the matrix associated with $\sigma^{1/2} - X_i$ is a Schur complement in  $\sigma^{1/2}$, $\sigma^{1/2} - X_i \ge 0$. Thus $\left\{ \Pi_i \right\}_{i=1}^{m}$ and $Z$ satisfy equation \eqref{dual5}. By theorem \ref{thm:optimality} this shows that the pair $\left( \left\{ \Pi_i \right\}_{i=1}^{m}, Z \right)$ is an optimal dual pair for the MED of $\mathrm{P}$.

Using the definition of $Z$ given in \eqref{ZQ}, and equations \eqref{dual1} and \eqref{PGMQ} we see that $p_i\rho_i$ should satisfy the following equation
\begin{equation}
\label{verify1}
\dfrac{p_i \rho_i \Pi_i p_i \rho_i}{\tr Z^2} = \dfrac{ Z \Pi_i Z}{\tr Z^2} = q_i \sigma_i, \; i\in \{1,\cdots,m\}.
\end{equation}
Hence $\mathrm{Range} \, q_i \sigma_i \subseteq \mathrm{Range} \, p_i \rho_i $, for all $i=1,\cdots,m$. But since  $\mathrm{Rank} \, q_i \sigma_i = \mathrm{Rank \ } p_i \rho_i = r_i$ for each $i=1,\cdots,m$, we get $\mathrm{Range} \, q_i \sigma_i = \mathrm{Range} \, p_i \rho_i $. Since the $\sigma_i$'s are linearly independent states, the $\rho_i$'s are also linearly independent states. This shows that $\mathrm{P} \in \ens$. From the construction, the $X_i$ are uniquely determined, and hence the map $\mathfrak{R}': \ens \longrightarrow \ens$ is well-defined and this completes the proof.\end{proof}

Hence in the theorem we show that

\begin{equation}
\label{op1}
\mathbf{OP}\left( \mathfrak{R}' \left( \mathrm{Q} \right) \right) = \mathrm{PGM} \left( \mathrm{Q} \right).
\end{equation}

We have shown that $\mathfrak{R}'$ is a well-defined map, and we will show that this map is actually the inverse of $\mathfrak{R}$. The map $\mathfrak{R}$ was defined using equations \eqref{sigma} and \eqref{qii}. We see from equation \eqref{verify1} that $ \mathfrak{R} \left( \mathrm{P} \right) = \mathrm{Q}$, and hence we get that  for each $\mathrm{Q}$ in $ \ens$,
\begin{equation}
\label{rightinverse}
\mathfrak{R} \circ \mathfrak{R}' \left( \mathrm{Q} \right) = \mathrm{Q}.
\end{equation}

To establish that $\mathfrak{R}'$ is the inverse of $\mathfrak{R}$, it remains to show the following.

\begin{thm}
\label{verify2}
$\mathfrak{R}'\circ \mathfrak{R} \left( \mathrm{P} \right) = \mathrm{P}$, for all $\mathrm{P} \in \ens$.
\end{thm}

\begin{proof}
For any $\mathrm{P} = \left\{ p_i, \rho_i \right\}_{i=1}^{m} \in \ens$, we obtain $\mathrm{Q} = \mathfrak{R} \left( \mathrm{P} \right) = \left\{ q_i, \sigma_i \right\}_{i=1}^{m}$, using equation \eqref{qisi}. Hence by Theorem \ref{thm5}, $\mathrm{PGM} \left( \mathrm{Q} \right) = \mathbf{OP} \left( \mathrm{P} \right).$ Let $\mathbf{OP}\left(\mathrm{P}\right)= \left\{ \Pi_i \right\}_{i=1}^{m}$. Thus by equation \eqref{dual2}, $Z = \sum_{i=1}^{m} p_i \rho_i \Pi_i = \sum_{i=1}^m p_i \Pi_i \rho_i$. By equation \eqref{sigmaZ}, we also have that $Z = \sqrt{\tr Z^2} \, \sigma^{1/2}$. Let $Z' = \dfrac{ \sigma^{1/2}}{\sum_{i=1}^{m} \tr X_j}$, where $Z'$ was introduced in equation \eqref{ZQ}. Thus $Z = c Z'$, where $c > 0$ is some constant. Let $\mathfrak{R}'\left( \mathrm{Q} \right) = \mathrm{P}' = \left\{ p'_i, \rho'_i \right\}_{i=1}^{m}$, where $p'_i$ and $\rho'_i$ were defined in equation \eqref{pi}. Then we obtain the following conclusions.
\begin{equation}\label{compare1}
\Pi_i p_i \rho_i \Pi_i = \Pi_i Z \Pi_i = c\Pi_i  Z'\Pi_i = c  \Pi_i p'_i \rho'_i \Pi_i,
\end{equation}

\begin{align}\label{compare2}
\Pi_i p_i \rho_i \left(\mathrm{Id}-\Pi_i\right) = \Pi_i Z \left(\mathrm{Id}-\Pi_i\right)
 & = c \Pi_i Z' \left(\mathrm{Id}-\Pi_i\right) \notag \\
&=   c \Pi_i p'_i \rho'_i \left(\mathrm{Id}-\Pi_i\right),
\end{align}
and
\begin{align}
\label{compare3}
\left( \mathrm{Id} - \Pi_i\right) p_i \rho_i  \Pi_i
 = \left( \mathrm{Id} - \Pi_i\right) Z \Pi_i
& =c \left(\mathrm{Id}-\Pi_i\right)Z' \Pi_i \notag \\
& = c\left( \mathrm{Id} - \Pi_i\right) p'_i \rho'_i \Pi_i.
\end{align}
Using equations \eqref{X_i} and \eqref{pi}, we may represent $p_i \rho_i$ in the same orthonormal basis used in equation \eqref{sigma12inbasis} as follows
$$p_i \rho_i\longleftrightarrow\dfrac{c}{\sum_{j=1}^{m} \tr \, X_j} \begin{pmatrix}
                                                   A_i & B_i \\
                                                   B_i^\dag & W_i
                                                  \end{pmatrix},$$
where $W_i$ an $(d-r_i)\times (d-r_i)$ matrix, which should be positive semidefinite.  $W_i - B_i^\dag A_i^{-1} B_i$ is the Schur complement of $A_i$ in $\left(\begin{smallmatrix} A_i & B_i \\                                               B_i^\dag & W_i \end{smallmatrix}\right)$. Using a result in \cite{Boyd},
\begin{equation}
\label{rhoranks}
\mathrm{Rank} \, p_i \rho_i = \mathrm{Rank} \, A_i + \mathrm{Rank} \, \left( W_i - B_i^\dag A_i^{-1} B_i\right),
\end{equation}
but since $\mathrm{Rank} \, p_i \rho_i = \mathrm{Rank} \, A_i = r_i$, we get that $\mathrm{Rank} \, \left( W_i - B_i^\dag A_i^{-1} B_i\right) = 0$. In other words, $W_i = B_i^\dag A_i^{-1} B_i$, and thus
\begin{equation}
\label{rhocompare}
p_i \rho_i\longleftrightarrow\dfrac{c}{\sum_{j=1}^{m} \tr \, X_j} \begin{pmatrix}
                                                   A_i & B_i \\
                                                   B_i^\dag & B^\dag A_i^{-1} B_i
                                                  \end{pmatrix},
\end{equation} hence $p_i \rho_i = c p'_i \rho'_i$. But note that $\sum_{i=1}^{m} \tr \, p_i \rho_i = c = 1$. Hence we obtain that $\mathrm{P} = \mathrm{P'}$. Thus $\mathfrak{R}'(\mathrm{Q})=\mathrm{P}$, and hence $\mathfrak{R}'\circ \mathfrak{R} \left( \mathrm{P} \right) = \mathrm{P}$, for all $\mathrm{P} \in \ens$.
\end{proof}

Thus we have proved that $\mathfrak{R}'$ is the left and right inverse of $\mathfrak{R}$, which implies that $\mathfrak{R}$ is a bijection. Also, note that we have explicitly constructed the mapping $\mathfrak{R}^{-1}$.

In the course of the proof of Theorem \ref{verify2}, we find a simplified condition for optimality { with respect to the one} given in Theorem \ref{thm:optimality}. We establish this below.


\begin{thm}
\label{simple1}
Let $\mathrm{P} = \left\{p_i, \rho_i \right\}_{i=1}^m \in \ens$. Then $\left\{ \Pi_i \right\}_{i=1}^m \in \pro$ is the optimal POVM for MED of $\mathrm{P}$ if { and only if}
\begin{itemize}
 \item[1.] $\left\{ \Pi_i\right\}_{i=1}^m$ satisfies equation \eqref{dual4} (or equivalently equation \eqref{dual1}) and
 \item[2.] $\sum_{j=1}^{m} p_j \rho_j \Pi_j { > } 0$.
\end{itemize}
\end{thm}

{ \begin{proof} First, let's assume that 1. and 2. are true. To prove that $\{ \Pi_i \}_{i=1}^{m}$ is the optimal POVM, we need to show that the inequality \eqref{dual5} is also true, i.e., we need to show that $\sum_{j=1}^{m} p_j \rho_j \Pi_j - p_i \rho_i \ge 0$, for all $ i \in \{1,2,\cdots,m\}$. To see this, choose an orthonormal basis in which $\Pi_i$ is diagonal. In this basis let $\sum_{j=1}^{m} p_j \rho_j \Pi_j $ have the matrix representation $\left( \begin{smallmatrix} A_i & B_i \\ B_i^\dag & C_i \end{smallmatrix}\right)$. Condition 2. implies that  $\left( \begin{smallmatrix} A_i & B_i \\ B_i^\dag & C_i \end{smallmatrix}\right) >0$. Note that $A_i$ is $r_i \times r_i$ and $A_i > 0$. Next we prove that $p_i \rho_i$ has the matrix representation given by $\left( \begin{smallmatrix} A_i & B_i \\ B_i^\dag & B_i^\dag A_i^{-1} B_i \end{smallmatrix} \right)$. To see this note that $$\Pi_i p_i \rho_i \Pi_i = \Pi_i \left( \sum_{j=1}^{m} p_j \rho_j \Pi_j \right) \Pi_i \leftrightarrow \left(\begin{matrix} A_i &  0 \\ 0 & 0 \end{matrix}\right) ,$$ $$\left(\mathrm{Id}-\Pi_i\right) p_i \rho_i \Pi_i = \left(\mathrm{Id}-\Pi_i\right) \left( \sum_{j=1}^{m} p_j \rho_j \Pi_j \right) \Pi_i \leftrightarrow \left(\begin{matrix} 0 &  0 \\ B_i^\dag & 0 \end{matrix}\right),$$ and $$\Pi_i p_i \rho_i \left(\mathrm{Id}-\Pi_i\right)=\Pi_i \left( \sum_{j=1}^{m} p_j \rho_j \Pi_j \right) \left(\mathrm{Id}-\Pi_i\right) \leftrightarrow \left(\begin{matrix} 0 &  B_i \\ 0 & 0 \end{matrix}\right).$$ Then the matrix representation of $p_i \rho_i$ is of the form $\left( \begin{smallmatrix} A_i & B_i \\ B_i^\dag & W_i \end{smallmatrix}\right)$. The Schur complement of $A_i$ in $\left( \begin{smallmatrix} A_i & B_i \\ B_i^\dag & W_i \end{smallmatrix}\right)$ is $W_i - B_i^\dag A_i^{-1} B_i$, and using the same reasoning employed between equations \eqref{rhoranks} and \eqref{rhocompare}, we get that $W_i = B_i^\dag A_i^{-1} B_i$. Thus $\sum_{j=1}^{m} p_j \rho_j \Pi_j - p_i \rho_i $ has the matrix representation $\left( \begin{smallmatrix} 0 & 0 \\ 0 & C_i - B_i^\dag A_i^{-1} B_i \end{smallmatrix} \right)$. Note that $C_i - B_i^\dag A_i^{-1} B_i $ is the Schur complement of $A_i$ in $\left( \begin{smallmatrix} A_i & B_i \\ B_i^\dag & C_i \end{smallmatrix}\right)$. Since $\left( \begin{smallmatrix} A_i & B_i \\ B_i^\dag & C_i \end{smallmatrix}\right) > 0$, the Schur complement of $A_i$ in $\left( \begin{smallmatrix} A_i & B_i \\ B_i^\dag & C_i \end{smallmatrix}\right)$ is also positive definite \cite{Boyd}. Thus $\left( \begin{smallmatrix} 0 & 0 \\ 0 & C_i - B_i^\dag A_i^{-1} B_i \end{smallmatrix} \right) \ge 0$, and hence  $\sum_{j=1}^{m} p_j \rho_j \Pi_j - p_i \rho_i  \ge 0$.

Conversely, let's assume that $\{ \Pi_i \}_{i=1}^m$ is the optimal POVM. Then $\{ \Pi_i \}_{i=1}^m$ must satisfy conditions \eqref{dual4} (condition \eqref{dual1}) and conditions  \eqref{dual5}. Thus $\sum_{j=1}^{m} p_j \rho_j \Pi_j \ge p_i \rho_i$, for all $ i \in \{1,2,\cdots,m\}$. Summing the LHS and RHS of this inequality over the index $i$ gives us $\sum_{j=1}^m p_j \rho_j \Pi_j \ge \frac{1}{m} \sum_{i=1}^m p_i \rho_i.$ Since the eigenvectors of the $\rho_i$'s span $\mathcal{H}$, we have that $\sum_{i=1}^{m} p_i \rho_i > 0$, and thus $\sum_{j=1}^m p_j \rho_j \Pi_j > 0$. 
 \end{proof}}

In fact, the optimality conditions can be simplified even further.

\begin{corollary}
\label{cor}
Let $\mathrm{P} \in \ens$ and $\left\{ \Pi_j \right\}_{j=1}^{m} \in \pro$. Then $\left\{ \Pi_j \right\}_{j=1}^{m}$ is the optimal POVM for the MED of $\mathrm{P}$ if { and only if} $\sum_{i=1}^{m} p_i \rho_i \Pi_i > 0$.
\end{corollary}
\begin{proof}
{ For sufficiency, we have to prove condition $(2)$ in Theorem \ref{simple1}. This was already proved in \cite{BC09}, but for completeness we still prove it here:} since $\sum_{i=1}^{m} p_i \rho_i \Pi_i > 0$, $\sum_{i=1}^{m} p_i \rho_i \Pi_i = \sum_{i=1}^{m} p_i  \Pi_i \rho_i$. Thus we have that
\begin{align}
\label{compact}
  & \Pi_j \left( \sum_{i=1}^{m} p_i  \Pi_i \rho_i - \sum_{i=1}^{m} p_i \rho_i \Pi_i \right) \Pi_k \notag \\ = & \Pi_j  \left( p_j \rho_j - p_k \rho_k \right)\Pi_k \notag  \\& = 0,
\end{align}
where we used the fact that $\left\{ \Pi_i \right\}_{i=1}^m$ is a projective measurement. Thus $\left\{ \Pi_i \right\}_{i=1}^m$ satisfy the condition \eqref{dual1}. Also, note that $\sum_{i=1}^{m} p_i \rho_i \Pi_i > 0$, so by Theorem \ref{simple1}, $\left\{ \Pi_j \right\}_{j=1}^{m}$ is the optimal POVM for the MED of $\mathrm{P}$.
{ For the necessity, assume that $\left\{ \Pi_j \right\}_{j=1}^{m}$ is the optimal POVM for the MED of $\mathrm{P}$. Then condition  ($2$) of Theorem \ref{simple1} is true.}
\end{proof}

Hence Theorem \ref{simple1} and Corollary \ref{cor} tell us that the optimality conditions for the MED of ensembles of LI states are actually simpler than for the case of more general ensembles of states. This also generalizes the results in \cite{Helstrom82, Pozza15}.

\section{Fixed Points of $\mathfrak{R}$}
\label{sec:fix}

Let $\mathrm{P} \in \ens$ be a fixed point of $\mathfrak{R}$, i.e., $\mathfrak{R} \left( \mathrm{P} \right) = \mathrm{P}$. Then by \eqref{mainresult} we have $$ \mathbf{OP} \left( \mathrm{P} \right) = \mathrm{PGM} \left( \mathrm{P} \right).$$ In other words,  if $\mathrm{P}$ is a fixed point of $\mathfrak{R}$, then its PGM is the optimal POVM. In the following theorem,  we give necessary and sufficient conditions for $\mathrm{P}$ to be a fixed point of $\mathfrak{R}$.

\begin{thm}
\label{fix}
Let $\mathrm{P} = \left\{ p_i, \rho_i \right\}_{i=1}^m$ be an element in $\ens$. Then $\mathfrak{R}(\mathrm{P})=\mathrm{P}$ if and only if $\sum_{i=1}^{m} \Pi_i \rho^{1/2} \Pi_i = c \, \mathrm{Id}$, for some constant $c>0$, where $\left\{ \Pi_i \right\}_{i=1}^{m} = \mathrm{PGM} \left( \mathrm{P} \right)$ and $\rho = \sum_{i=1}^{m} p_i \rho_i$.
\end{thm}
\begin{proof}
Suppose that $\sum_{i=1}^{m} \Pi_i \rho^{1/2} \Pi_i = c \, \mathrm{Id}$, for some constant $c>0$, where $\left\{ \Pi_i \right\}_{i=1}^{m} = \mathrm{PGM} \left( \mathrm{P} \right)$ and $\rho = \sum_{i=1}^{m} p_i \rho_i$. Then for each $ i = 1,2 \cdots m$,
\begin{equation}\label{constant}
\Pi_i \rho^{1/2} \Pi_i = c \Pi_i.
\end{equation}
Let $\mathfrak{R}^{-1}\left( \mathrm{P} \right) = \mathrm{P}' = \left\{ p'_i, \rho'_i \right\}_{i=1}^m$. By \eqref{op1}, $\mathbf{OP}(\mathrm{P}')=\mathrm{PGM}(\mathrm{P})=\{\Pi_i\}_{i=1}^m$ and from equation \eqref{sigmaZ} we get the optimal dual pair  $\left(  \left\{ \Pi_i \right\}_{i=1}^m, t \rho^{1/2} \right)$  for MED of $\mathrm{P}'$, where $t >0$ is some constant. Now we follow the same sequence of steps as in proof of Theorem \ref{verify2} to show that $\mathfrak{R}^{-1}\left( \mathrm{P} \right) = \mathrm{P}$ by using the relation \eqref{constant}

Let us fix an orthonormal basis which diagonalizes $\Pi_i$ and we use this basis to obtain matrix representations.
Consider the following matrix representation of  $t \rho^{1/2}$;
\begin{align}
 \label{reprho}
 t \rho^{1/2} \longleftrightarrow \begin{pmatrix}  A_i & B_i \\ B_i^\dag & D_i \end{pmatrix},
 \end{align}
where $A_i$ represents $t \Pi_i  \rho^{1/2} \Pi_i$. Then by \eqref{constant}, $A_i=tc \ \mathrm{Id}_{r_i}$. By the optimality conditions \eqref{dual1} we have $t \rho^{1/2}\Pi_i = p'_i \rho'_i \Pi_i$ and thus the matrix representation of $p_i'\rho_i'$ is given by
\begin{equation}
\label{last1}
 p'_i \rho'_i  \longleftrightarrow \begin{pmatrix} tc \mathrm{Id}_{r_i} & B_i \\ B_i^\dag &  \frac{1}{tc} B_i^\dag B_i \end{pmatrix},
\end{equation}
\noindent where $\frac{1}{tc} B_i^\dag B_i$ is obtained from equation \eqref{rhocompare} with   $A_i^{-1}=\frac{1}{tc} \ \mathrm{Id}_{r_i}$. Note that from equation \eqref{qisi}, $p_i \rho_i = \dfrac{{p'_i}^2 \rho'_i \Pi_i \rho'_i}{t^2}$, which has the following matrix representation
\begin{align}
\label{last2}
 & p_i \rho_i  \longleftrightarrow \notag \\ \; & \dfrac{1}{t^2} \begin{pmatrix} tc \mathrm{Id}_{r_i} & B_i \\ B_i^\dag &  \frac{1}{tc} B_i^\dag B_i \end{pmatrix} \begin{pmatrix} \mathrm{Id_{r_i}} & 0 \\ 0 & 0 \end{pmatrix}\begin{pmatrix} tc \mathrm{Id}_{r_i} & B_i \\ B_i^\dag &  \frac{1}{tc} B_i^\dag B_i \end{pmatrix} \notag \\
 = \;  & \frac{c}{t} \begin{pmatrix} tc \mathrm{Id}_{r_i} & B_i \\ B_i^\dag &  \frac{1}{tc} B_i^\dag B_i \end{pmatrix}.
\end{align}
\noindent Comparing equations \eqref{last1} and \eqref{last2} we get that $p_i \rho_i = (c/t) p'_i \rho'_i$. Summing over $i$ and taking trace gives us that $c = t$. Thus $p_i \rho_i = p'_i \rho'_i$, for all $ 1 \le i \le m$. Thus we get that $\mathrm{P}' = \mathrm{P}$, or that $\mathfrak{R}\left( \mathrm{P} \right) = \mathrm{P}$.

Conversely, for some $\mathrm{P}=\{p_i,\rho_i\}_{i=1}^m\in \ens$, let $\mathfrak{R}(\mathrm{P})=\mathrm{P}$.
Then $\mathbf{OP} \left( \mathrm{P} \right) = \mathrm{PGM} \left( \mathrm{P} \right) =\left\{ \Pi_i \right\}_{i=1}^{m}$. Let $(\left\{ \Pi_i \right\}_{i=1}^{m}, Z)$ be the optimal dual pair for MED of $P$. Then by \eqref{sigmaZ} and \eqref{dual2}, we have, for some constant $c>0$,
\begin{equation}\label{zzz}
Z=c\rho^{1/2}=\sum_{i=1}^m p_i\rho_i \Pi_i
\end{equation}
Since $\Pi_i=\rho^{-1/2}p_i\rho_i\rho^{-1/2}$ with $\rho=\sum_{i=1}^mp_i\rho_i$, we get $p_i\rho_i=\rho^{1/2}\Pi_i\rho^{1/2}$. Thus we have
$p_i \rho_i \Pi_i = \rho^{1/2} \Pi_i \rho^{1/2} \Pi_i$. Then by \eqref{zzz},
\[c\rho^{1/2}=Z=\sum_{i=1}^m p_i \rho_i \Pi_i=\rho^{1/2}\sum_{i=1}^m\Pi_i\rho^{1/2}\Pi_i\]
and hence $\sum_{i=1}^m\Pi_i\rho^{1/2}\Pi_i=c \mathrm{Id}$.
\end{proof}

Theorem \ref{fix} tells us that the PGM is the optimal POVM when the probability of successfully identifying the $i$-th state is proportional to $\mathrm{Rank} \, \rho_i$, i.e., $p_i \tr \Pi_i \rho_i \propto r_i$, for all $ 1 \le i \le m$. In \cite{Masahide98} it was shown that when the states $\rho_i$ are LI and pure, i.e., $\rho_i \longrightarrow \ket{\psi_i}\bra{\psi_i}$ and the $\ket{\psi_i}$'s are LI, then the PGM is the optimal POVM when the probability of successfully identifying the $i$-th state is independent of $i$, i.e. $p_i \bra{\psi_i} \Pi_i \ket{\psi_i} = c$, for some constant $c > 0$. Hence Theorem \ref{fix} reduces to the result in \cite{Masahide98} for the case of linearly independent pure state ensembles.

\section{Discussion and Conclusion}
\label{sec:last}

In this work we generalize the results for the MED problem of LI pure state ensembles to mixed state ensembles. Firstly, we show that there exists a map $\mathfrak{R}$ on the set of LI ensembles, such that the pretty good measurement of the image of this map is the optimal POVM for the MED of the pre-image.  Next, we show that $\mathfrak{R}$ is bijective, and we explicitly construct $\mathfrak{R}^{-1}$. This generalizes results obtained in \cite{Mochon2006}. The fixed points of $\mathfrak{R}$ are seen to be ensembles whose pretty good measurements are optimal for MED. In Theorem \ref{fix} we obtain necessary and sufficient conditions for an ensemble to be a fixed point of $\mathfrak{R}$. It is seen that for such cases, the probability of successfully detecting the $i$-th state is proportional to the rank of that state for all $1 \le i \le m$. This generalizes the result for LI pure state ensembles in \cite{Masahide98}, where it was shown that the probability of successfully detecting the $i$-th state is independent of $i$. Also, in Theorem \ref{simple1} and Corollary \ref{cor} we show that the optimality conditions for the MED of LI states is in fact simpler than the optimality conditions for general ensembles of states. This generalizes a result in obtained in \cite{Helstrom82, Pozza15}.

While the geometric structure of the MED problem \cite{Bae2013} has been employed to study it, particularly for the case of qubit systems \cite{Ha13,Ha15,Bae13}, the structure which Belavkin introduced in \cite{Belavkin75} has received scant attention. In \cite{Mochon2006}, Mochon rediscovered the structure for the case of pure state ensembles, and proved the existence of the map $\mathfrak{R}$ for LI pure states ensembles. This map was later employed in \cite{Singal2016} to obtain the optimal POVM. Equations \eqref{mainresult} tells us that to solve the MED problem it suffices to know the map $\mathfrak{R}$. However the construction of $\mathfrak{R}$ requires the optimal POVM. In fact it is a difficult problem to get an exact form of $\mathfrak{R}$. On the other hand, we have constructed $\mathfrak{R}^{-1}$ and thus if one can invert $\mathfrak{R}^{-1}$, then one solves the MED problem. This was done for the case of LI pure state ensembles in \cite{Singal2016}, where the authors used the implicit function theorem to do so. We would like to see if this can be generalized to the case of LI mixed state ensembles as well. Work for this is under progress.

\section*{Acknowledgments} 

We thank the referee for their valuable comments.

\end{document}